\documentclass[smallextended]{svjour2}

\usepackage{amsmath, amstext, amssymb, amsfonts, graphicx, bm, subfigure}

\newcommand{\PR}{\mathbb P}
\newcommand{\G}{\mathcal G}

\journalname{Journal of Statistical Physics}

\begin{document}

\title{A detailed investigation into near degenerate exponential random graphs}

\author{Mei Yin \thanks{Mei Yin's research was partially
supported by NSF grant DMS-1308333.}}

\institute{Mei Yin \at Department of Mathematics, University of
Denver, Denver, CO 80208, USA \\
              \email{mei.yin@du.edu}\\
}

\date{Received: date / Accepted: date}

\maketitle

\begin{abstract}
The exponential family of random graphs has been a topic of continued research interest. Despite the relative simplicity, these models capture a variety of interesting features displayed by large-scale networks and allow us to better understand how phases transition between one another as tuning parameters vary. As the parameters cross certain lines, the model asymptotically transitions from a very sparse graph to a very dense graph, completely skipping all intermediate structures. We delve deeper into this near degenerate tendency and give an explicit characterization of the asymptotic graph structure as a function of the parameters.

\keywords{Exponential random graphs \and Phase transitions \and Near degeneracy}

\subclass{05C80 \and 82B26}
\end{abstract}

\section{Introduction}
\label{intro}
Exponential random graphs represent an important and challenging class of models, displaying both diverse and novel phase transition phenomena. These rather general models are exponential families of probability distributions over graphs, in which dependence between the random edges is defined through certain finite subgraphs, in imitation of the use of potential energy to provide dependence between particle states in a grand canonical ensemble of statistical physics. They are particularly useful when one wants to construct models that resemble observed networks as closely as possible, but without specifying an explicit network formation mechanism. Consider the set $\mathcal{G}_n$ of all simple graphs
$G_n$ on $n$ vertices (``simple'' means undirected, with no loops
or multiple edges). By a $k$-parameter family of exponential
random graphs we mean a family of probability measures
$\PR_n^{\beta}$ on $\G_n$ defined by, for $G_n\in\G_n$,
\begin{equation}
\label{pmf} \PR_n^{\beta}(G_n)=\exp\left(n^2\left(\beta_1
t(H_1,G_n)+\cdots+
  \beta_k t(H_k,G_n)-\psi_n^{\beta}\right)\right),
\end{equation}
where $\beta=(\beta_1,\dots,\beta_k)$ are $k$ real parameters,
$H_1,\dots,H_k$ are pre-chosen finite simple graphs (and we take
$H_1$ to be a single edge), $t(H_i, G_n)$ is the density of graph
homomorphisms (the probability that a random vertex map $V(H_i)
\to V(G_n)$ is edge-preserving),
\begin{equation}
\label{t} t(H_i, G_n)=\frac{|\text{hom}(H_i,
G_n)|}{|V(G_n)|^{|V(H_i)|}},
\end{equation}
and $\psi_n^{\beta}$ is the normalization constant,
\begin{equation}
\label{psi} \psi_n^{\beta}=\frac{1}{n^2}\log\sum_{G_n \in
\mathcal{G}_n} \exp\left(n^2 \left(\beta_1
t(H_1,G_n)+\cdots+\beta_k t(H_k,G_n)\right) \right).
\end{equation}
Intuitively, we can think of the $k$ parameters $\beta_1, \ldots, \beta_k$ as tuning parameters that allow one to adjust the influence of different subgraphs $H_1, \ldots, H_k$ of $G_n$ on the probability distribution and analyze the extent to which specific values of these subgraph densities ``interfere'' with one another. Since the real-world networks the exponential models depict are often very large in size, our main interest lies in exploring the structure of a typical graph drawn from the model when $n$ is large.

This subject has attracted enormous attention in mathematics, as well as in various applied disciplines. Many of the investigations employ the elegant theory of graph limits as developed by Lov\'{a}sz and coauthors (V.T. S\'{o}s, B. Szegedy, C. Borgs, J. Chayes, K. Vesztergombi, \ldots) \cite{BCLSV1} \cite{BCLSV2} \cite{BCLSV3} \cite{Lov} \cite{LS}. Building on earlier work of Aldous \cite{Aldous1} and Hoover \cite{Hoover}, the graph limit theory creates a new set of tools for representing and studying the asymptotic behavior of graphs by connecting sequences of graphs to a unified graphon space equipped with a cut metric. Though the theory itself is tailored to dense graphs, serious attempts have been made at formulating parallel results for sparse graphs \cite{AL} \cite{BS} \cite{BCCZ1} \cite{BCCZ2} \cite{CD2} \cite{LZ2}. Applying the graph limit theory to $k$-parameter exponential random graphs and utilizing a large deviations result for Erd\H{o}s-R\'{e}nyi graphs established in Chatterjee and Varadhan \cite{CV}, Chatterjee and Diaconis \cite{CD1} showed that when $n$ is large and $\beta_2,\ldots,\beta_k$ are non-negative, a typical graph drawn from the ``attractive'' exponential model (\ref{pmf}) looks like an Erd\H{o}s-R\'{e}nyi graph $G(n, u^*)$ in the graphon sense, where the edge presence probability $u^*(\beta_1,\ldots,\beta_k)$ is picked randomly from the set of maximizers of a variational problem for the limiting normalization constant $\psi_\infty^\beta=\lim_{n\to \infty}\psi_n^{\beta}$:
\begin{equation}
\label{max}
\psi_{\infty}^{\beta}=\sup_{0\leq u\leq
1}\left(\beta_1 u^{e(H_1)}+\cdots+\beta_k
u^{e(H_k)}-\frac{1}{2}u\log u-\frac{1}{2}(1-u)\log(1-u)\right),
\end{equation}
where $e(H_i)$ is the number of edges in $H_i$. They also noted that in the edge-(multiple)-star model where $H_j$ is a $j$-star for $j=1,\ldots,k$, due to the unique structure of stars, maximizers of the variational problem for $\psi_\infty^\beta$ for all parameter values satisfy (\ref{max}) and a typical graph drawn from the model is always Erd\H{o}s-R\'{e}nyi. Since the limiting normalization constant is the generating function for the limiting expectations of other random variables on the graph space such as expectations and correlations of homomorphism densities, these crucial observations connect the occurrence of an asymptotic phase transition in (\ref{pmf}) with an abrupt change in the solution of (\ref{max}) and have led to further exploration into exponential random graph models and their variations.

Being exponential families with finite support, one might expect exponential random graph models to enjoy a rather basic asymptotic form, though in fact, virtually all these models are highly nonstandard as the size of the network increases. The $2$-parameter exponential random graph models have therefore generated continued research interest. These prototype models are simpler than their $k$-parameter extensions but nevertheless exhibit a wealth of non-trivial attributes and capture a variety of interesting features displayed by large networks. The relative simplicity furthermore helps us better understand how phases transition between one another as tuning parameters vary and provides insight into the expressive power of the exponential construction. In the statistical physics literature, phase transition is often associated with loss of analyticity in the normalization constant, which gives rise to discontinuities in the observed statistics. For exponential random graph models, phase transition is characterized as a sharp, unambiguous partition of parameter ranges separating those values in which changes in parameters lead to smooth changes in the homomorphism densities, from those special parameter values where the response in the densities is singular.

For the ``attractive'' $2$-parameter edge-triangle model obtained by taking $H_1=K_2$ (an edge), $H_2=K_3$ (a triangle), and $\beta_2\geq 0$, Chatterjee and Diaconis \cite{CD1} gave the first rigorous proof of asymptotic singular behavior and identified a curve $\beta_2=q(\beta_1)$ across which the model transitions from very sparse to very dense, completely skipping all intermediate structures. In further works (see for example, Radin and Yin \cite{RY}, Aristoff and Zhu \cite{AZ1}), this singular behavior was discovered universally in generic $2$-parameter models where $H_1$ is an edge and $H_2$ is any finite simple graph, and the transition curve $\beta_2=q(\beta_1)$ asymptotically approaches the straight line $\beta_2=-\beta_1$ as the parameters diverge. The double asymptotic framework of \cite{CD1} was later extended in \cite{YRF}, and the scenario in which the parameters diverge along general straight lines $\beta_1=a\beta_2$, where $a$ is a constant and $\beta_2 \rightarrow \infty$, was considered. Consistent with the near degeneracy predictions in \cite{AZ1} \cite{CD1} \cite{RY}, asymptotically for $a\leq -1$, a typical graph sampled from the ``attractive'' $2$-parameter exponential model is sparse, while for $a>-1$, a typical graph is nearly complete. Although much effort has been focused on $2$-parameter models, $k$-parameter models have also been examined. As shown in \cite{Yin1}, near degeneracy and universality are expected not only in generic $2$-parameter models but also in generic $k$-parameter models. Asymptotically, a typical graph drawn from the ``attractive'' $k$-parameter exponential model where $\beta_2,\ldots, \beta_k \geq 0$ is sparse below the hyperplane $\sum_{i=1}^k \beta_i=0$ and nearly complete above it. For the edge-(multiple)-star model, the desirable star feature relates to network expansiveness and has made predictions of similar asymptotic phenomena possible in broader parameter regions. Related results may be found in H\"{a}ggstr\"{o}m and Jonasson \cite{HJ}, Park and Newman \cite{PN}, Bianconi \cite{Bianconi}, Lubetzky and Zhao \cite{LZ1}, Radin and Sadun \cite{RS1} \cite{RS2}, and Kenyon et al. \cite{KRRS}.

These theoretical findings have advanced our understanding of
phase transitions in exponential random graph models, yet some
important questions remain unanswered. Previous investigations
identified near degenerate parameter regions in which a typical
sampled graph looks like an Erd\H{o}s-R\'{e}nyi graph $G(n, u^*)$,
where the edge presence probability $u^*\rightarrow 0$ or $1$, but
the speed of $u^*$ towards these two degenerate states is not at
all clear. When a typical graph is sparse ($u^*\rightarrow 0$),
how sparse is it? When a typical graph is nearly complete
$(u^*\rightarrow 1)$, how dense is it? Can we give an explicit
characterization of the near degenerate graph structure as a
function of the parameters? The rest of this paper is dedicated
towards these goals. Some of the ideas for the sparse case were
partially implemented in \cite{YZ1}. Theorem \ref{2generic}
considers generic ``attractive'' $2$-parameter exponential random
graph models and Theorem \ref{singlestar} derives parallel results
for ``repulsive'' edge-(single)-star models. The asymptotic
characterization of $u^*$ obtained in these theorems then make
possible a deeper exploration into the asymptotics of the limiting
normalization constant of the exponential model in Theorem
\ref{norm}, which indicates that though a typical graph displays
Erd\H{o}s-R\'{e}nyi feature, the simplified Erd\H{o}s-R\'{e}nyi
graph and the real exponential graph are not exact asymptotic
analogs in the usual statistical physics sense. In the sparse
region, the Erd\H{o}s-R\'{e}nyi graph does seem to reflect the
asymptotic tendency of the exponential graph more accurately, as
the two interpretations of the limiting normalization constant
coincide when the parameters diverge. Lastly, Theorems
\ref{generic} and \ref{inf} further extend the near degenerate
analysis from $2$-parameter exponential random graph models to
$k$-parameter exponential random graph models.

\section{Investigating near degeneracy}
\label{deg}
This section explores the exact asymptotics of generic $2$-parameter exponential random graph models where $\beta_2\geq 0$ near degeneracy. The analysis is then extended to $\beta_2<0$ for the edge-(single)-star model. By including only two subgraph densities in the exponent,
\begin{equation}
\label{2pmf} \PR_n^{\beta}(G_n)=\exp\left(n^2\left(\beta_1
t(H_1,G_n)+\beta_2 t(H_2,G_n)-\psi_n^{\beta}\right)\right),
\end{equation}
where $H_1$ is an edge and $H_2$ is a different finite simple graph, the $2$-parameter models are arguably simpler than their $k$-parameter generalizations. As illustrated in Chatterjee and Diaconis \cite{CD1}, when $n$ is large and $\beta_2$ is non-negative, a typical graph drawn from the ``attractive'' $2$-parameter exponential model (\ref{2pmf}) behaves like an Erd\H{o}s-R\'{e}nyi graph $G(n, u^*)$, where the edge presence probability $u^*(\beta_1,\beta_2)$ is picked randomly from the set of maximizers of a variational problem for the limiting normalization constant $\displaystyle \psi^\beta_\infty$:
\begin{equation}
\label{reduce}
\psi_\infty^\beta=\sup_{0\leq u\leq 1}\left(\beta_1 u+\beta_2 u^p-\frac{1}{2}u\log u-\frac{1}{2}(1-u)\log(1-u)\right),
\end{equation}
where $p$ is the number of edges in $H_2$, and thus satisfies
\begin{equation}
\label{sd}
\beta_1+\beta_2p(u^*)^{p-1}=\frac{1}{2}\log\left(\frac{u^*}{1-u^*}\right).
\end{equation}
This implicitly describes $u^*$ as a function of the parameters $\beta_1$ and $\beta_2$, but a closed-form solution is not obtainable except when $\beta_2=0$, which gives $u^*(\beta_1,0)=e^{2\beta_1}/(1+e^{2\beta_1})$. For $\beta_1$ large negative, $u^*(\beta_1,0)$ asymptotically behaves like $e^{2\beta_1}$, while for $\beta_1$ large positive, $u^*(\beta_1,0)$ asymptotically behaves like $1-e^{-2\beta_1}$. We would like to know if these asymptotic results could be generalized. By \cite{YRF}, taking $\beta_1=a\beta_2$ and $\beta_2\rightarrow \infty$, $u^* \rightarrow 0$ when $a\leq -1$ and $u^* \rightarrow 1$ when $a>-1$. Equivalently, for $(\beta_1,\beta_2)$ sufficiently far away from the origin, $u^* \rightarrow 0$ when $\beta_1\leq -\beta_2$ and $u^* \rightarrow 1$ when $\beta_1>-\beta_2$. As regards the speed of $u^*$ towards these two degenerate states, simulation results suggest that $u^*$ is asymptotically $e^{2\beta_1}$ in the former case and is asymptotically $1-e^{-2(\beta_1+p\beta_2)}$ in the latter case. See Tables \ref{table1} and \ref{table2} and Figure \ref{figure1}. Even for $\beta$ with small magnitude, the asymptotic tendency of $u^*$ is quite evident.

\begin{table}
\begin{center}
\begin{tabular}{ccccc}
$\beta_1$ & $\beta_2$ & $u^*(\beta_1,\beta_2)$ & $e^{2\beta_1}$ & $1-e^{-2(\beta_1+3\beta_2)}$ \\
\hline \hline \\
$-2$ & $1$ & $0.01802$ & $0.01832$  & \\
$-2$ & $2$ & $0.01806$ & $0.01832$ & \\
$0$ & $1$ & $0.99745$ & & $0.99752$ \\
$2$ & $1$ & $0.99995$ & & $0.99995$ \\ \\
\end{tabular}
\end{center}
\caption{Asymptotic comparison for ``attractive'' edge-triangle model near degeneracy.} \label{table1}
\end{table}

\begin{theorem}
\label{2generic}
Consider an ``attractive'' $2$-parameter exponential random graph model (\ref{2pmf}) where $\beta_2\geq 0$. For large $n$ and $(\beta_1,\beta_2)$ sufficiently far away from the origin, a typical graph drawn from the model looks like an Erd\H{o}s-R\'{e}nyi graph $G(n, u^*)$, where the edge presence probability $u^*(\beta_1,\beta_2)$ satisfies:
\begin{itemize}
\item $u^* \asymp e^{2\beta_1}$ if $\beta_1\leq -\beta_2$,
\item $u^* \asymp 1-e^{-2(\beta_1+p\beta_2)}$ if $\beta_1>-\beta_2$.
\end{itemize}
\end{theorem}

\begin{proof}
As explained in the previous paragraph, in the large $n$ limit, the asymptotic edge presence probability $u^*(\beta_1,\beta_2)$ of a typical sampled graph is prescribed according to the maximization problem (\ref{sd}). For $(\beta_1,\beta_2)$ whose magnitude is sufficiently big, $u^* \rightarrow 0$ when $\beta_1\leq -\beta_2$ and $u^* \rightarrow 1$ when $\beta_1>-\beta_2$.

For $\beta_1\leq -\beta_2$, we rewrite (\ref{sd}) in the following way:
\begin{equation}
\frac{u^*}{e^{2\beta_1+2\beta_2p(u^*)^{p-1}}}=1-u^*.
\end{equation}
Take $0<\epsilon<1/p$, since $u^* \rightarrow 0$, for $(\beta_1,\beta_2)$ sufficiently far away from the origin, $(u^*)^{p-1}<\epsilon$. Using $1-u^*\leq 1$, we then have
\begin{equation}
u^*\leq e^{2\beta_1+2\beta_2p\epsilon}\leq e^{2\beta_2(-1+p\epsilon)}.
\end{equation}
This implies that
\begin{equation}
2\beta_2p(u^*)^{p-1}\leq 2\beta_2p e^{2\beta_2(p-1)(-1+p\epsilon)} \rightarrow 0
\end{equation}
as $\beta_2$ gets sufficiently large. Using $1-u^* \rightarrow 1$ again, this further shows that $u^*$ asymptotically behaves like $e^{2\beta_1}$.

\begin{figure}
\centering
\includegraphics[width=3in]{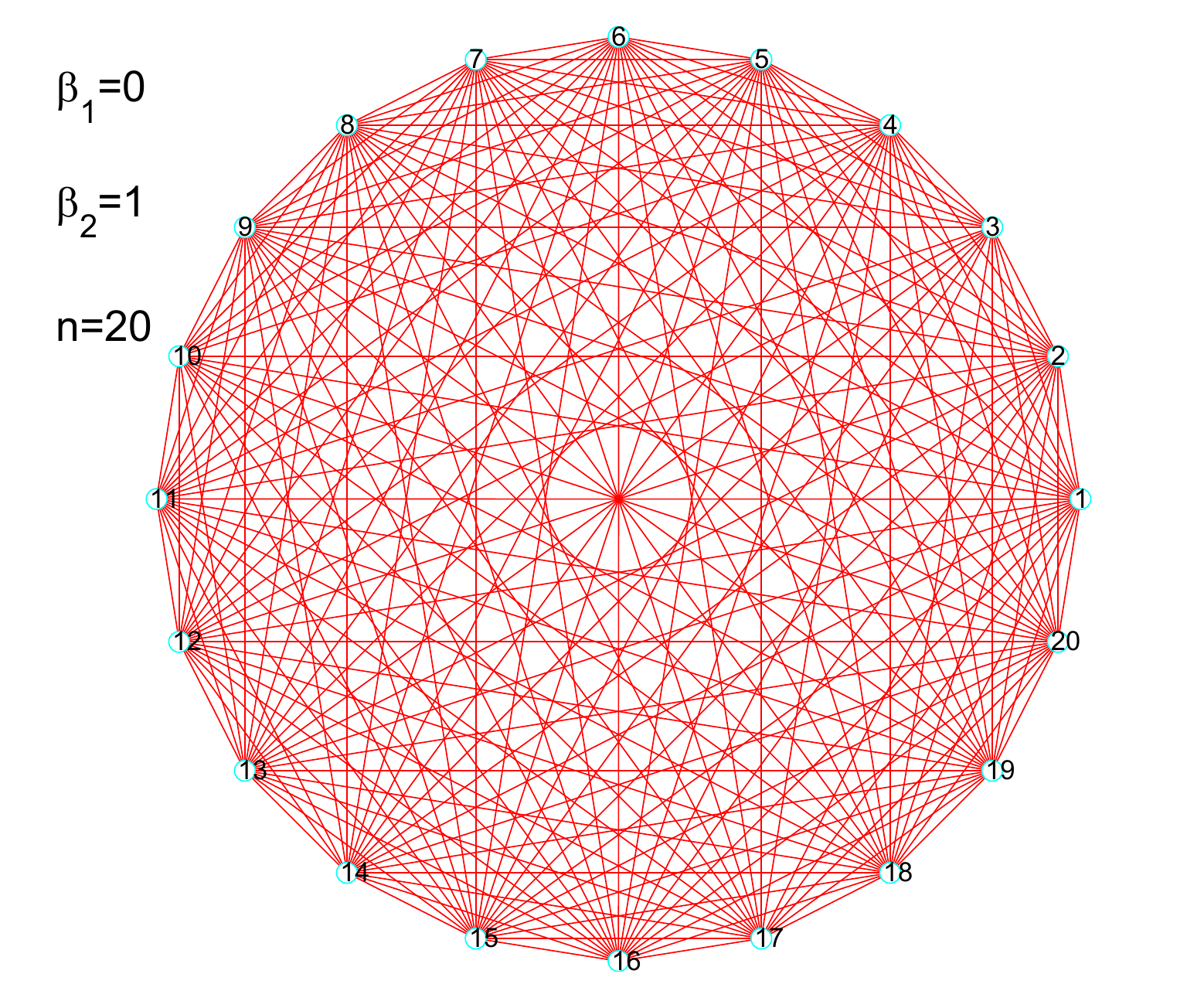}
\caption{A simulated realization of the exponential random graph model on $20$ nodes with edges and triangles as sufficient statistics, where $\beta_1=0$ and $\beta_2=1$. The simulated graph displays Erd\H{o}s-R\'{e}nyi feature with edge density $0.995$, matching the asymptotic predictions of Theorem \ref{2generic}. (cf. Table \ref{table1}: $u^*(\beta_1,\beta_2)=0.99745$ and $1-e^{-2(\beta_1+3\beta_2)}=0.99752$).}
\label{figure1}
\end{figure}

For $\beta_1>-\beta_2$, we rewrite (\ref{sd}) in the following way:
\begin{equation}
\frac{v^*}{e^{-2\beta_1-2\beta_2p(1-v^*)^{p-1}}}=1-v^*, \end{equation}
where $v^*=1-u^* \rightarrow 0$. Going one step further, we separate $1-(p-1)v^*$ from $(1-v^*)^{p-1}$:
\begin{equation}
\label{exp}
\frac{v^*}{e^{-2\beta_1-2\beta_2p+2\beta_2p(p-1)v^*}}=(1-v^*)e^{-2\beta_2p\sum_{s=2}^{p-1}{\binom {p-1}{s}}(-v^*)^s} \leq 1-v^*, \end{equation}
as the dominating term in the exponent $-2\beta_2p{\binom {p-1}{2}}(-v^*)^2$ carries a negative sign. Take $0<\epsilon<1/p$, since $v^* \rightarrow 0$, for $(\beta_1,\beta_2)$ sufficiently far away from the origin, $v^*<\epsilon$. Using $1-v^*\leq 1$, we then have
\begin{equation}
v^*\leq e^{-2\beta_1-2\beta_2p+2\beta_2p(p-1)\epsilon}\leq e^{2\beta_2(1-p+p(p-1)\epsilon)}=e^{2\beta_2(p-1)(-1+p\epsilon)}.
\end{equation}
This implies that
\begin{equation}
2\beta_2p(p-1)v^*\leq 2\beta_2p(p-1) e^{2\beta_2(p-1)(-1+p\epsilon)} \rightarrow 0 \end{equation}
as $\beta_2$ gets sufficiently large, and since $v^* \rightarrow 0$, also implies that the sum of all the terms in the exponent $-2\beta_2p\sum_{s=2}^{p-1}{\binom {p-1}{s}}(-v^*)^s \rightarrow 0$. Using $1-v^* \rightarrow 1$ again, this further shows that $v^*$ asymptotically behaves like $e^{-2(\beta_1+p\beta_2)}$, or equivalently, $u^*$ asymptotically behaves like $1-e^{-2(\beta_1+p\beta_2)}$.
\end{proof}

\begin{figure}
\centering
\includegraphics[width=4in]{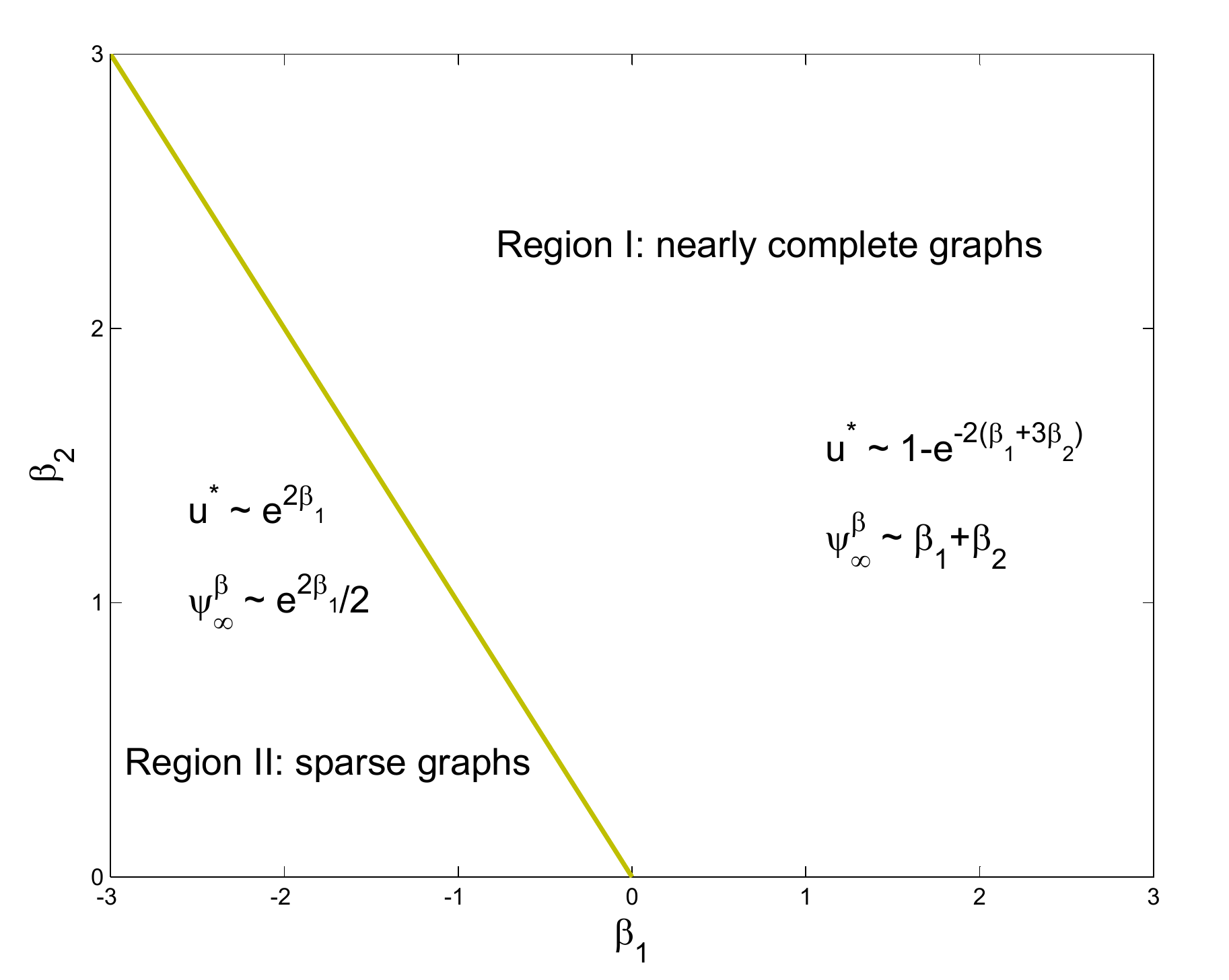}
\caption{Asymptotic tendency in ``attractive'' edge-triangle model.} \label{triangle}
\end{figure}

In the edge-(single)-star model where $H_2$ is a star with $p$ edges, due to the unique structure of stars, maximizers of the variational problem for the limiting normalization constant $\psi_\infty^\beta$ when the parameter $\beta_2<0$ again satisfies (\ref{sd}), and the near degeneracy predictions in Theorem \ref{2generic} may be extended from the upper half-plane to the lower half-plane. It was shown in \cite{YZ1} that for large $n$ and $(\beta_1,\beta_2)$ sufficiently far away from the origin, a typical graph drawn from the ``repulsive'' edge-(single)-star model where $\beta_2<0$ is indistinguishable from an Erd\H{o}s-R\'{e}nyi graph $G(n, u^*)$, where the edge presence probability $u^* \rightarrow 0$ when $\beta_1< 0$ and $u^* \rightarrow 1$ when $\beta_1> -p\beta_2$. As regards the speed of $u^*$ towards these two degenerate states, simulation results suggest that just as in the ``attractive'' situation, $u^*$ is asymptotically $e^{2\beta_1}$ in the sparse case and is asymptotically $1-e^{-2(\beta_1+p\beta_2)}$ in the nearly complete case. See Table \ref{table2}. Even for $\beta$ with small magnitude, the asymptotic tendency of $u^*$ is quite evident.

\begin{theorem}
\label{singlestar}
Consider a ``repulsive'' edge-(single)-star model obtained by taking $H_2$ a star with $p$ edges and $\beta_2<0$ in (\ref{2pmf}). For large $n$ and $(\beta_1,\beta_2)$ sufficiently far away from the origin, a typical graph drawn from the model looks like an Erd\H{o}s-R\'{e}nyi graph $G(n, u^*)$, where the edge presence probability $u^*(\beta_1,\beta_2)$ satisfies:
\begin{itemize}
\item $u^* \asymp e^{2\beta_1}$ if $\beta_1<0$,
\item $u^* \asymp 1-e^{-2(\beta_1+p\beta_2)}$ if $\beta_1>-p\beta_2$.
\end{itemize}
\end{theorem}

\begin{proof}
For $(\beta_1,\beta_2)$ whose magnitude is sufficiently big, we examine the maximization problem (\ref{sd}) separately when $\beta_1<0$ and when $\beta_1>-p\beta_2$.

First for $\beta_1<0$. Assume that $\beta_1\leq \delta_1 \beta_2$ for some fixed but arbitrary $\delta_1>0$. We rewrite (\ref{sd}) in the following way:
\begin{equation}
\frac{u^*}{e^{2\beta_1}}=(1-u^*)e^{2\beta_2p(u^*)^{p-1}}\leq 1-u^*.
\end{equation}
Using $1-u^*\leq 1$, we then have
\begin{equation}
u^*\leq e^{2\beta_1} \leq e^{2\delta_1\beta_2}.
\end{equation}
This implies that
\begin{equation}
2|\beta_2|p(u^*)^{p-1}\leq 2|\beta_2|pe^{2\delta_1\beta_2(p-1)} \rightarrow 0
\end{equation}
as $\beta_2$ gets sufficiently negative. Using $1-u^* \rightarrow 1$ again, this further shows that $u^*$ asymptotically behaves like $e^{2\beta_1}$.

\begin{table}
\begin{center}
\begin{tabular}{ccccc}
$\beta_1$ & $\beta_2$ & $u^*(\beta_1,\beta_2)$ & $e^{2\beta_1}$ & $1-e^{-2(\beta_1+2\beta_2)}$ \\
\hline \hline \\
$-2$ & $-4$ & $0.01435$ & 0.01832 & \\

$-2$ & $-2$ & $0.01588$ & $0.01832$ & \\

$-3$ & $1$ & $0.00250$ & $0.00248$ & \\

$-3$ & $3$ & $0.00255$ & $0.00248$ & \\

$1$ & $1$ & $0.99750$ & & $0.99752$ \\

$4$ & $-1$ & $0.98317$ & & $0.98168$ \\ \\
\end{tabular}
\end{center}
\caption{Asymptotic comparison for edge-$2$-star model near degeneracy.} \label{table2}
\end{table}

Next for $\beta_1>-p\beta_2$. Assume that $\beta_1\geq -(p+\delta_2)\beta_2$ for some fixed but arbitrary $\delta_2>0$. We rewrite (\ref{sd}) in the following way:
\begin{equation}
\frac{v^*}{e^{-2\beta_1-2\beta_2p(1-v^*)^{p-1}}}=1-v^*,
\end{equation}
where $v^*=1-u^* \rightarrow 0$. Going one step further, we separate $1$ from $(1-v^*)^{p-1}$:
\begin{equation}
\frac{v^*}{e^{-2\beta_1-2\beta_2p}}=(1-v^*)e^{-2\beta_2p\sum_{s=1}^{p-1}{\binom {p-1}{s}}(-v^*)^s} \leq 1-v^*,
\end{equation}
as the dominating term in the exponent $2\beta_2p(p-1)v^*$ carries a negative sign. Using $1-v^*\leq 1$, we then have
\begin{equation}
v^*\leq e^{-2\beta_1-2\beta_2p}\leq e^{2\delta_2\beta_2}.
\end{equation}
This implies that
\begin{equation}
2|\beta_2|p(p-1)v^*\leq 2|\beta_2|p(p-1)e^{2\delta_2\beta_2}\rightarrow 0
\end{equation}
as $\beta_2$ gets sufficiently negative, and since $v^* \rightarrow 0$, also implies that the sum of all the terms in the exponent $-2\beta_2p\sum_{s=1}^{p-1}{\binom {p-1}{s}}(-v^*)^s \rightarrow 0$. Using $1-v^* \rightarrow 1$ again, this further shows that $v^*$ asymptotically behaves like $e^{-2(\beta_1+p\beta_2)}$, or equivalently, $u^*$ asymptotically behaves like $1-e^{-2(\beta_1+p\beta_2)}$.
\end{proof}

Though the $2$-parameter exponential random graph $G_n$ looks like an Erd\H{o}s-R\'{e}nyi random graph $G(n, u^*)$ in the large $n$ limit, we also note some marked dissimilarities. The limiting normalization constant $\psi_\infty^\beta$ for the $2$-parameter exponential model (\ref{2pmf}) is given by (\ref{reduce}), while the ``equivalent'' Erd\H{o}s-R\'{e}nyi model yields that $\psi_\infty^\beta$ is $-\log(1-u^*)/2$. Since $u^*$ is nonzero for finite $(\beta_1,\beta_2)$ \cite{CD1}, the two different interpretations of the limiting normalization constant indicate that the simplified Erd\H{o}s-R\'{e}nyi graph and the real exponential model are not exact asymptotic analogs in the usual statistical physics sense. When the relevant Erd\H{o}s-R\'{e}nyi graph is near degenerate, Theorems \ref{2generic} and \ref{singlestar} give the asymptotic speed of $u^*$ as a function of $\beta_1$ and $\beta_2$, allowing a deeper analysis of the asymptotics of $\psi_\infty^\beta$ in the following Theorem \ref{norm}. The theorem is formulated in the context of the edge-(single)-star model, since the asymptotics of $u^*$ are known in broader parameter regions for this model, but the statement for the ``attractive'' situation ($\beta_2\geq 0$) applies without modification to generic $2$-parameter models. See Figures \ref{triangle} and \ref{star}. We also note that, in the sparse region, the Erd\H{o}s-R\'{e}nyi graph seems to reflect the asymptotic tendency of the exponential random graph more accurately, as the two interpretations of the limiting normalization constant do coincide when the parameters diverge.

\begin{figure}
\centering
\includegraphics[width=4in]{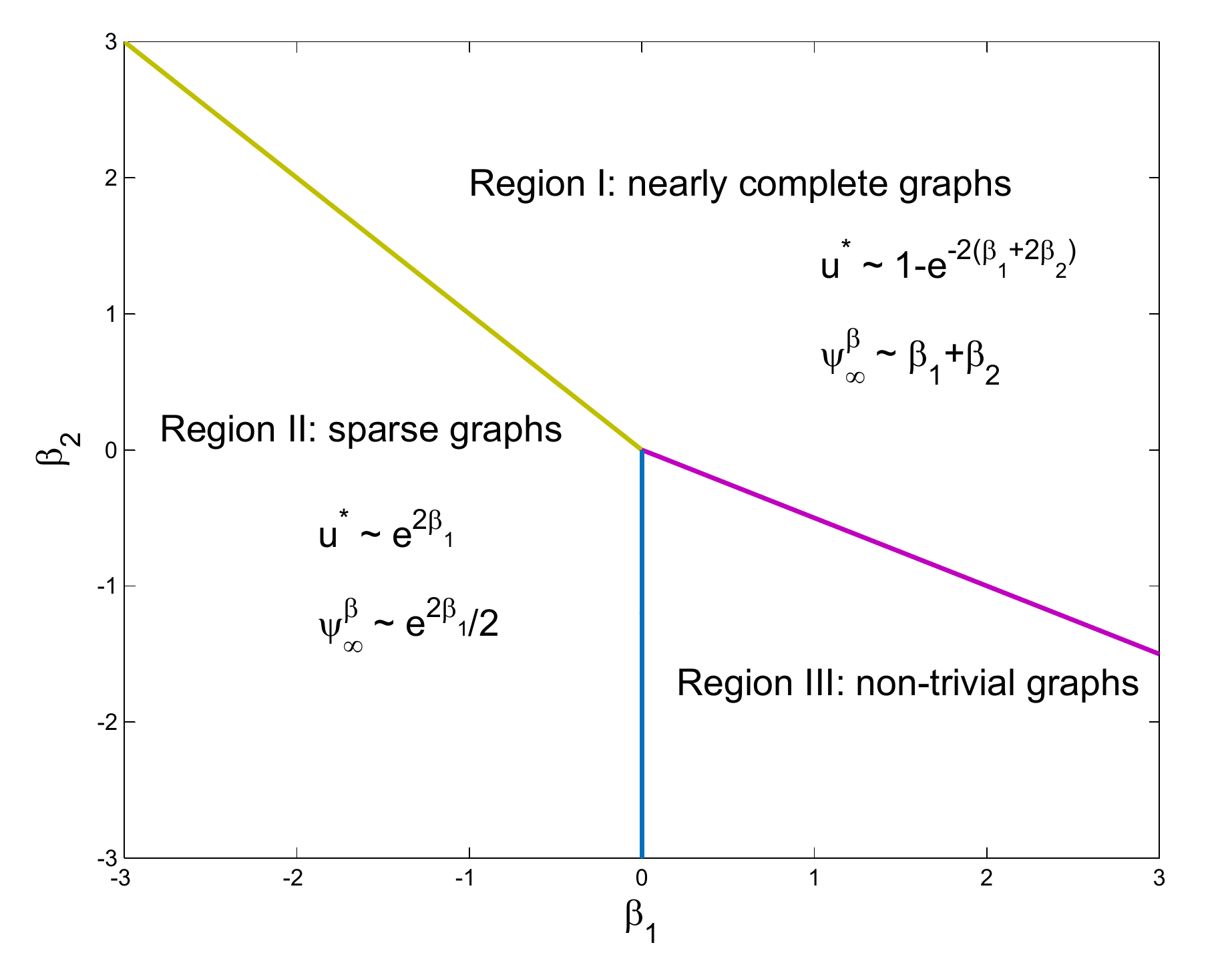}
\caption{Asymptotic tendency in edge-$2$-star model.} \label{star}
\end{figure}

\begin{theorem}
\label{norm}
Consider an edge-(single)-star model obtained by taking $H_2$ a star with $p$ edges in (\ref{2pmf}). For $(\beta_1,\beta_2)$ sufficiently far away from the origin, the limiting normalization constant $\psi_\infty^\beta$ satisfies:
\begin{itemize}
\item $\psi_\infty^\beta \asymp e^{2\beta_1}/2$ if $\beta_1\leq -\beta_2$ and $\beta_1<0$,
\item $\psi_\infty^\beta \asymp \beta_1+\beta_2$ if $\beta_1>-\beta_2$ and $\beta_1>-p\beta_2$.
\end{itemize}
\end{theorem}

\begin{proof}
For $(\beta_1,\beta_2)$ whose magnitude is sufficiently big, we examine the limiting normalization constant (\ref{reduce}) separately in the sparse region and in the nearly complete region.

In the sparse region ($\beta_1\leq -\beta_2$ and $\beta_1<0$),
\begin{eqnarray}
\psi_\infty^\beta&=&\beta_1u^*+\beta_2(u^*)^p-\frac{1}{2}u^*\log u^*-\frac{1}{2}(1-u^*)\log(1-u^*)\\
&=&\beta_2(u^*)^p(1-p)-\frac{1}{2}\log(1-u^*)\notag.
\end{eqnarray}
From Theorems \ref{2generic} and \ref{singlestar}, $\beta_2(u^*)^{p-1}\rightarrow 0$ and $-\log(1-u^*)\asymp u^* \asymp e^{2\beta_1}$. This shows that $\psi_\infty^\beta$ asymptotically behaves like $e^{2\beta_1}/2$.

In the nearly complete region ($\beta_1>-\beta_2$ and $\beta_1>-p\beta_2$),
\begin{eqnarray}
\psi_\infty^\beta&=&\beta_1(1-v^*)+\beta_2(1-v^*)^p-\frac{1}{2}v^*\log v^*-\frac{1}{2}(1-v^*)\log(1-v^*)\\
&=&\beta_2(1-v^*)^p(1-p)-\frac{1}{2}\log v^*\notag.
\end{eqnarray}
From Theorems \ref{2generic} and \ref{singlestar}, $\beta_2v^* \rightarrow 0$ and $v^* \asymp e^{-2(\beta_1+p\beta_2)}$. This shows that $\psi_\infty^\beta$ asymptotically behaves like $\beta_1+\beta_2$.
\end{proof}

We may also analyze the asymptotics of $\psi_\infty^\beta$ along the boundaries of the near degenerate region. The boundary of the sparse region is given by $\beta_1=0$ and $\beta_2<0$, and $u^*$ satisfies
\begin{equation}
\beta_2p(u^*)^{p-1}=\frac{1}{2}\log\left(\frac{u^*}{1-u^*}\right).
\end{equation}
Though $u^* \rightarrow 0$ depends on $\beta_2$ in a complicated way, the asymptotic behavior of $\psi_\infty^\beta$ can be characterized:
\begin{equation}
\psi_\infty^\beta=\frac{1-p}{2p}u^*\log\frac{u^*}{1-u^*}-\frac{1}{2}\log(1-u^*).
\end{equation}
Using $\log(1-u^*)\asymp -u^*$, this shows that $\psi_\infty^\beta$ asymptotically behaves like $(1-p)u^*\log u^*/(2p)$. We recognize that the asymptotic behaviors of $\psi_\infty^\beta$ on the boundary of and inside the sparse region are different: Inside, $\psi_\infty^\beta$ is asymptotically $u^*/2$ and converges to $0$, while on the boundary, $\psi_\infty^\beta$ though also converges to $0$ is at a much slower rate.

The boundary of the nearly complete region is given by $\beta_1=-p\beta_2$ and $\beta_2<0$, and $u^*$ satisfies
\begin{equation}
\label{comp}
-\beta_2p+\beta_2p(1-v^*)^{p-1}=-\frac{1}{2}\log\left(\frac{v^*}{1-v^*}\right).
\end{equation}
Though $v^*=1-u^*\rightarrow 0$ depends on $\beta_2$ in a complicated way, the asymptotic behavior of $\psi_\infty^\beta$ can be characterized:
\begin{equation}
\psi_\infty^\beta=\beta_2(1-p)(1-v^*)-\frac{1-p}{2p}(1-v^*)\log\left(\frac{v^*}{1-v^*}\right)-\frac{1}{2}\log v^*.
\end{equation}
Since the dominating term on the left of (\ref{comp}) is $-\beta_2p(p-1)v^*$, using $\log(1-v^*)\rightarrow 0$, we then have $\beta_2p(p-1)v^*\asymp \log v^*/2$, which shows that $|\beta_2|$ is asymptotically larger compared with $|\log v^*|$ and further shows that $\psi_\infty^\beta$ asymptotically behaves like $\beta_2(1-p)$. We recognize that the asymptotic behaviors of $\psi_\infty^\beta$ on the boundary of and inside the nearly complete region coincide.

\section{Further discussion}
\label{discuss}
This section extends the investigation into near degeneracy from generic $2$-parameter exponential random graph models to generic $k$-parameter exponential random graph models. For ``attractive'' models where $\beta_2,\ldots,\beta_k\geq 0$, we derive parallel results concerning the asymptotic graph structure and the limiting normalization constant. Using similar methods, more results can be deduced for the ``repulsive'' edge-(multiple)-star model where $\beta_2,\ldots,\beta_k<0$. As illustrated in Chatterjee and Diaconis \cite{CD1}, when $n$ is large and $\beta_2,\ldots,\beta_k$ are non-negative, a typical graph drawn from the $k$-parameter exponential model behaves like an Erd\H{o}s-R\'{e}nyi graph $G(n, u^*)$, where the edge presence probability $u^*(\beta_1,\ldots,\beta_k)$ is picked randomly from the set of maximizers of (\ref{max}), and thus satisfies
\begin{equation}
\label{kreduce}
\beta_1+\beta_2e(H_2)(u^*)^{e(H_2)-1}+\cdots+\beta_ke(H_k)(u^*)^{e(H_k)-1}=\frac{1}{2}\log\left(\frac{u^*}{1-u^*}\right),
\end{equation}
where $e(H_i)$ is the number of edges in $H_i$. We take $H_1$ to be an edge and assume without loss of generality that $1=e(H_1)\leq \cdots\leq e(H_k)$.

\begin{theorem}
\label{generic}
Consider an ``attractive'' $k$-parameter exponential random graph model (\ref{pmf}) where $\beta_2,\ldots,\beta_k\geq 0$. For large $n$ and $(\beta_1,\ldots,\beta_k)$ sufficiently far away from the origin, a typical graph drawn from the model looks like an Erd\H{o}s-R\'{e}nyi graph $G(n, u^*)$, where the edge presence probability $u^*(\beta_1,\ldots,\beta_k)$ satisfies:
\begin{itemize}
\item $u^* \asymp e^{2\beta_1}$ if $\sum_{i=1}^k \beta_i\leq 0$,
\item $u^* \asymp 1-e^{\sum_{i=1}^k -2\beta_ie(H_i)}$ if $\sum_{i=1}^k \beta_i>0$.
\end{itemize}
\end{theorem}

\begin{proof}
The proof follows a similar line of reasoning as in the proof of Theorem \ref{2generic}. Expectedly though, the argument is more involved since we are working with $k$-parameter families rather than $2$-parameter families.

For $\sum_{i=1}^k \beta_i \leq 0$, we rewrite (\ref{kreduce}) in the following way:
\begin{equation}
\frac{u^*}{e^{2\beta_1+\sum_{i=2}^k 2\beta_ie(H_i)(u^*)^{e(H_i)-1}}}=1-u^*.
\end{equation}
Take $0<\epsilon<1/e(H_k)$, since $u^*\rightarrow 0$, for $(\beta_1,\ldots,\beta_k)$ sufficiently far away from the origin, $(u^*)^{e(H_i)-1}<\epsilon$ for $2\leq i\leq k$. Using $1-u^*\leq 1$, we then have
\begin{equation}
u^*\leq e^{2\beta_1+\sum_{i=2}^k 2\beta_ie(H_i)\epsilon}\leq e^{\sum_{i=2}^k 2\beta_i(-1+e(H_i)\epsilon)}. \end{equation}
This implies that
\begin{equation}
2\beta_je(H_j)(u^*)^{e(H_j)-1}\leq 2\beta_j e(H_j) e^{(e(H_j)-1)\sum_{i=2}^k 2\beta_i(-1+e(H_i)\epsilon)} \rightarrow 0
\end{equation}
for all $2\leq j\leq k$ as $\beta_2,\ldots,\beta_k$ get sufficiently large. Using $1-u^*\rightarrow 0$ again, this further shows that $u^*$ asymptotically behaves like $e^{2\beta_1}$.

For $\sum_{i=1}^k \beta_i>0$, we rewrite (\ref{kreduce}) in the following way:
\begin{equation}
\frac{v^*}{e^{-2\beta_1-\sum_{i=2}^k 2\beta_ie(H_i)(1-v^*)^{e(H_i)-1}}}=1-v^*,
\end{equation}
where $v^*=1-u^*\rightarrow 0$. Going one step further, for $2\leq i\leq k$, we separate $1-(e(H_i)-1)v^*$ from $(1-v^*)^{e(H_i)-1}$:
\begin{multline}
\frac{v^*}{e^{\sum_{i=1}^k -2\beta_ie(H_i)+\sum_{i=2}^k 2\beta_ie(H_i)(e(H_i)-1)v^*}}\\=(1-v^*)e^{-\sum_{i=2}^k 2\beta_i e(H_i)\sum_{s=2}^{e(H_i)-1}\binom {e(H_i)-1}{s}(-v^*)^s}\leq 1-v^*,
\end{multline}
as the dominating term in the exponent $-\sum_{i=2}^k 2\beta_i e(H_i)\binom {e(H_i)-1}{2}(-v^*)^2$ carries a negative sign. Take $0<\epsilon<1/e(H_k)$, since $v^*\rightarrow 0$, for $(\beta_1,\ldots,\beta_k)$ sufficiently far away from the origin, $v^*<\epsilon$. Using $1-v^*\leq 1$, we then have
\begin{multline}
v^*\leq e^{\sum_{i=1}^k -2\beta_ie(H_i)+\sum_{i=2}^k 2\beta_ie(H_i)(e(H_i)-1)\epsilon}\\
\leq e^{\sum_{i=2}^k 2\beta_i(1-e(H_i)+e(H_i)(e(H_i)-1)\epsilon)}=e^{\sum_{i=2}^k 2\beta_i(e(H_i)-1)(-1+e(H_i)\epsilon)}.
\end{multline}
This implies that
\begin{equation}
2\beta_j e(H_j)(e(H_j)-1)v^*\leq 2\beta_j e(H_j)(e(H_j)-1)e^{\sum_{i=2}^k 2\beta_i(e(H_i)-1)(-1+e(H_i)\epsilon)} \rightarrow 0
\end{equation}
for all $2\leq j\leq k$ as $\beta_2,\ldots,\beta_k$ get sufficiently large, and since $v^*\rightarrow 0$, also implies that the sum of all the terms in the exponent\\ $-\sum_{i=2}^k 2\beta_i e(H_i)\sum_{s=2}^{e(H_i)-1}\binom {e(H_i)-1}{s}(-v^*)^s\rightarrow 0$. Using $1-v^*\rightarrow 1$ again, this further shows that $v^*$ asymptotically behaves like $e^{\sum_{i=1}^k -2\beta_ie(H_i)}$, or equivalently, $u^*$ asymptotically behaves like $1-e^{\sum_{i=1}^k -2\beta_ie(H_i)}$.
\end{proof}

\begin{theorem}
\label{inf}
Consider an ``attractive'' $k$-parameter exponential random graph model (\ref{pmf}) where $\beta_2,\ldots,\beta_k\geq 0$. For $(\beta_1,\ldots,\beta_k)$ sufficiently far away from the origin, the limiting normalization constant $\psi_\infty^\beta$ satisfies:
\begin{itemize}
\item $\psi_\infty^\beta \asymp e^{2\beta_1}/2$ if $\sum_{i=1}^k \beta_i \leq 0$,
\item $\psi_\infty^\beta \asymp \sum_{i=1}^k \beta_i$ if $\sum_{i=1}^k \beta_i>0$.
\end{itemize}
\end{theorem}

\begin{proof}
For $(\beta_1,\ldots,\beta_k)$ whose magnitude is sufficiently big, we examine the limiting normalization constant (\ref{max}) separately in the sparse region and in the nearly complete region.

In the sparse region ($\sum_{i=1}^k \beta_i\leq 0$),
\begin{eqnarray}
\psi_\infty^\beta&=&\sum_{i=1}^k \beta_i (u^*)^{e(H_i)}-\frac{1}{2}u^*\log u^*-\frac{1}{2}(1-u^*)\log(1-u^*)\\
&=& \sum_{i=2}^k \beta_i (u^*)^{e(H_i)}(1-e(H_i))-\frac{1}{2}\log(1-u^*)\notag.
\end{eqnarray}
From Theorem \ref{generic}, $\beta_i(u^*)^{e(H_i)-1}\rightarrow 0$ for all $2\leq i\leq k$ and $-\log(1-u^*)\asymp u^* \asymp e^{2\beta_1}$. This shows that $\psi_\infty^\beta$ asymptotically behaves like $e^{2\beta_1}/2$.

In the nearly complete region ($\sum_{i=1}^k \beta_i>0$),
\begin{eqnarray}
\psi_\infty^\beta&=& \sum_{i=1}^k \beta_i (1-v^*)^{e(H_i)}-\frac{1}{2}v^*\log v^*-\frac{1}{2}(1-v^*)\log(1-v^*)\\
&=& \sum_{i=2}^k \beta_i (1-v^*)^{e(H_i)}(1-e(H_i))-\frac{1}{2}\log v^*\notag.
\end{eqnarray}
From Theorem \ref{generic}, $\beta_i v^* \rightarrow 0$ for all $2\leq i\leq k$ and $v^* \asymp e^{\sum_{i=1}^k -2\beta_ie(H_i)}$. This shows that $\psi_\infty^\beta$ asymptotically behaves like $\sum_{i=1}^k \beta_i$.
\end{proof}

In the edge-(multiple)-star model, due to the unique structure of stars, maximizers of the variational problem for the limiting normalization constant $\psi_\infty^\beta$ satisfies (\ref{kreduce}) for any $\beta_1,\ldots,\beta_k$, and the near degeneracy predictions may be extended to the ``repulsive'' region. Using similar techniques as in \cite{YZ1}, it may be shown that for $(\beta_1,\ldots,\beta_k)$ sufficiently far away from the origin and $\beta_2,\ldots,\beta_k$ all negative, $u^*\rightarrow 0$ when $\beta_1<0$ and $u^*\rightarrow 1$ when $\sum_{i=1}^k \beta_ie(H_i)>0$. Then analogous conclusions as in Theorems \ref{generic} and \ref{inf} may be drawn:
\begin{itemize}
\item $u^* \asymp e^{2\beta_1}$ and $\psi_\infty^\beta \asymp e^{2\beta_1}/2$ if $\sum_{i=1}^k \beta_i \leq 0$ and $\beta_1<0$,
\item $u^* \asymp 1-e^{\sum_{i=1}^k -2\beta_ie(H_i)}$ and $\psi_\infty^\beta \asymp \sum_{i=1}^k \beta_i$ \\ \indent if $\sum_{i=1}^k \beta_i>0$ and $\sum_{i=1}^k \beta_ie(H_i)>0$.
\end{itemize}
We omit the proof details.

\section*{Acknowledgements}
The author is very grateful to the anonymous referees for the
invaluable suggestions that greatly improved the quality of this
paper. She appreciated the opportunity to talk about this work in
the Special Session on Topics in Probability at the 2016 AMS
Western Spring Sectional Meeting, organized by Tom Alberts and
Arjun Krishnan.

\end{document}